\documentclass[10pt,twocolumn,twoside]{IEEEtran}
\usepackage{cite}

\ifCLASSINFOpdf
   \usepackage[pdftex]{graphicx}
\else
\fi
\usepackage[cmex10]{amsmath}
\usepackage{amssymb}
\usepackage{amsfonts}
\usepackage{amsthm}
\usepackage{mathtools}

\usepackage[english]{babel}
\newtheoremstyle{note}
{3pt}
{3pt}
{}
{}
{\itshape}
{:}
{.2em}
{}
\newtheorem{theorem}{Theorem}
\newtheorem{lemma}{Lemma}

\newtheorem{definition}{Definition}
\newtheorem{remark}{Remark}
\newtheorem{problem}{Problem}
\usepackage{diagrams}
\usepackage{accents}
\usepackage{color}

\newcommand\munderbar[1]{%
  \underaccent{\bar}{#1}}
\begin{document}
%
\title{A Formal Traffic Characterization of LTI Event-triggered Control Systems}
%
%
%

\author{A.~Sharifi~Kolarijani and
        M.~Mazo~Jr.,~\IEEEmembership{Member,~IEEE}
\thanks{A.~Sharifi~ Kolarijani and M.~Mazo~Jr. are with the Delft Center for Systems and Control, Delft University of Technology, The Netherlands, e-mail: (\{a.sharifikolarijani, m.mazo\}@tudelft.nl).}
\thanks{Manuscript received April 00, 2015; revised December 00, 2015.}}

%
%

\markboth{Submitted}%
{Sharifi Kolarijani \MakeLowercase{\textit{et al.}}: A Formal Traffic Characterization of LTI Event-triggered Control Systems}
%



\maketitle

\begin{abstract}
Unnecessary communication and computation in the periodic execution of control tasks lead to over-provisioning in hardware design (or underexploitation in hardware utilization) in control applications, such as networked control systems. To address these issues, researchers have proposed a new class of strategies, named event-driven strategies. Despite of their beneficiary effects, matters like task scheduling and appropriate dimensioning of communication components have become more complicated with respect to traditional periodic strategies. In this paper, we present a formal approach to derive an abstracted system that captures the sampling behavior of a family of event-triggered strategies for the case of LTI systems. This structure approximately simulates the sampling behavior of the aperiodic control system. Furthermore, the resulting quotient system is equivalent to a timed automaton. In the construction of the abstraction, the state space is confined to a finite number of convex regions, each of which represents a mode in the quotient system. An LMI-based technique is deployed to derive a sampling time interval associated to each region. Finally, reachability analysis is leveraged to find the transitions of the quotient system.
\end{abstract}


%
\IEEEpeerreviewmaketitle

\section{Introduction}
\label{sec_2}
%
%
%
%
\IEEEPARstart{I}{N} networked control systems, particularly over wireless or shared channels, the scarcity of communication resources makes the application of traditional control strategies with periodic sampling inefficient. Alternative approaches with aperiodic sampling, such as event-triggered control (ETC) and self-triggered control (STC), have been recently proposed to reduce network usage. The underlying idea in these approaches is to take into account the dynamics of control systems in the sampling procedure, to only use the communication channel when strictly necessary.  

In these approaches, control actions are executed at the instants that a pre-specified condition, the so-called \emph{triggering mechanism}, is violated. In ETC, an intelligent sensory system continuously monitors the states of the control plant and determines the instants when the triggering condition is violated \cite{EDC}. ETC has been also used under other names, such as interrupt-based feedback control \cite{Hri}, Lebesgue sampling \cite{Ast}, asynchronous sampling \cite{Vou}, state-triggered feedback \cite{Tab}, and level crossing sampling \cite{Kof}. The necessity of an intelligent sensory system in ETC motivated researchers to propose another class of aperiodic approaches, namely STC \cite{Vel}. In STC, the controller is responsible for the determination of sampling instants, i.e., the next sampling instant is computed at the current sampling instant by the controller \cite{Anta, Man, Wang}. However, the robustness of STC strategies against disturbances is a concern because of the open-loop nature of the control action implementation between two consecutive sampling instants (notice that this is not an issue in ETC strategies which continuously supervise the plant). The assumption of availability of states in the feedback path is addressed in an extension to output-feedback event-triggered strategies in \cite{Don}. Most of the aforementioned approaches focus in decreasing the (computation and communication) resource utilization of real-time control systems. Nonetheless, there is a lack of tools to translate such approaches to frameworks that can be exploited by computer engineers in designing real-time systems. Providing such frameworks shall enable the reduction of hardware over-provisioning in the design process.

Several controller/scheduler co-design approaches of real-time systems can be found in the literature, e.g. feedback modification to task attributes \cite{But, Cac, Lu, Cer}, anytime controllers \cite{Bhat, Font}, and event-based control and scheduling \cite{Areq1, Areq2}. In \cite{But, Cac, Lu, Cer}, the principles of feedback control  theory were used to integrate scheduling and control in real-time system designs. \cite{Bhat} and \cite{Font} deployed the concept of anytime controllers in their co-design approaches in which there is a trade off between control performance and computation complexity. The problems of resource utilization and resource distribution have been jointly addressed in \cite{Areq1, Areq2}. These approaches include designing of a control law, a scheduler, and an event generator. The control law enhances the performance while the scheduler and event generator improve the efficiency of the resource usage. 
In the present paper we take an alternative approach providing a decoupling between controller design, event-triggered implementation and the scheduling design. Such decoupling between control and scheduling has the benefits of increased scalability and versatility of networked controller designs, with respect to monolithic co-design approaches as those mentioned earlier.

While in traditional periodic implementations of controllers the decoupling between control and scheduling is naturally provided by the period defining the frequency at which the control loop must be updated, in ETC and STC the matter is more involved. 
The triggering mechanism in ETC results in aperiodic control systems, in which sampling is a function of time and/or states, and thus no single parameter provides the necessary information for appropriate scheduling. In this study, we take one step further to complement a family of ETC strategies by abstracting their sampling behavior into simple structures that can be employed by real-time engineers for scheduling (and possibly network dimensioning). The objective of this paper is thus to model the \emph{traffic} generated by aperiodic ETC control systems, understood as all possible traces of sampling times. 
We achieve this by mapping the initial control system (which is infinite-state) into a type of quotient system (which is finite-state) called \emph{power quotient system}. We show that the power quotient system is in fact a \emph{timed automaton}~\cite{Alu}, and that it captures the original timing behavior accurately. The latter is shown by establishing an approximate simulation relation between the abstraction and the original system.

We consider linear time invariant (LTI) systems with state feedback laws and sample-and-hold control actions implemented in an ETC fashion as in~\cite{Tab_etc}. Our procedural technique is inspired by a state-dependent sampling proposed by \cite{Fit} where the $\beta$-stability \cite{Khal} of the system's origin is guaranteed based on an LMI-based approach derived from Lyapunov-Razumikhin stability conditions \cite{Kolm}. These conditions are developed based on the delayed nature of the system (related to the sample-and-hold nature of such systems) suggested by \cite{Frid}. Inspired by \cite{Fit}, we employ a two-step approach to compute sampling intervals associated to states: first, to remove the spatial dependencies, the state space is partitioned (abstracted) to a finite number of convex polyhedral cones (pointed at the origin); then, to remove the temporal dependencies, each conic region is associated to a time interval by using a convex embedding approach proposed by \cite{Het}. Each of these time intervals captures all possible inter-sample times of the corresponding region. To derive the desired quotient system, each conic region is considered as a discrete state (mode). Then, in order to derive the transitions of the quotient system, a reachability analysis is performed employing the sampling intervals computed earlier. This analysis is an extension of an approach proposed by \cite{Chut} for continuous autonomous systems, namely outer approximations of reachable sets by a set of convex polytopes. 
Finally, we show that the derived quotient systems are in fact equivalent to a \emph{timed safety automata}~\cite{Henz94} (\emph{timed automata}~\cite{Alu} without accepting conditions).
%

The organization of the remainder is as follows. The mathematical preliminaries and problem setup are presented in Section \ref{sec_3}. The formal methodology to construct the timing abstraction is explained in Section \ref{sec_4}. In Section \ref{sec_5}, the main results of this paper are summarized. Finally, an illustrative example is introduced in Section \ref{sec_6} and the paper is concluded with a brief discussion in Section \ref{sec_7}.

\section{Preliminaries and Problem Setup}
\label{sec_3}
In this section, we provide first the mathematical notation and revisit definitions of some notions employed in the paper. Next, we reformulate as a quadratic function of states at sampling instants the Input-to-State-Stability (ISS) based ETC approach proposed by \cite{Tab_etc} for LTI systems. Finally, the problem of deriving a structure that captures the sampling behavior of the ETC system is formalized. 
 
\subsection{Mathematical Preliminaries}
In what follows, $\mathbb{R}^n$ denotes the $n$-dimensional Euclidean space, $\mathbb{R}^+$ and $\mathbb{R}_0^+$ denote the positive and nonnegative reals, respectively, $\mathbb{N}$ is the set of positive integers, and $\mathbb{IR}^+$ is the set of all closed intervals $[a,b]$ such that $a,b \in \mathbb{R}^+$ and $a \leq b$. For any set $S$, $2^S$ denotes the set of all subsets of $S$, i.e. the power set of $S$. $\mathcal{M}_{m \times n}$ and $\mathcal{M}_n$ are the set of all $m \times n$ real-valued matrices and the set of all $n \times n$ real-valued symmetric matrices, respectively. For a matrix $M$, $M \preceq 0$ (or $M \succeq 0$) means $M$ is a negative (or positive) semidefinite matrix and $M \succ 0$ indicates $M$ is a positive definite matrix. $\lfloor x \rfloor$ indicates the largest integer not greater than $x \in \mathbb{R}$ and $|y|$ denotes the Euclidean norm of a vector $y \in \mathbb{R}^n$. Given two sets $Z_a$ and $Z_b$, every relation $Q \subseteq Z_a \times Z_b$ admits $Q^{-1}=\{(z_b,z_a)\in Z_b \times Z_a|(z_a,z_b) \in Q \}$ as its inverse relation. When $Q\subseteq Z\times Z$ is an equivalence relation on a set $Z$, $[z]$ denotes the equivalence class of $z \in Z$ and $Z/Q$ denotes the set of all equivalence classes. 

A fundamental observation that we employ in what follows is that the ordered pair $(\mathbb{IR}^+,d_H)$ is a metric space:

\begin{definition}{(Metric \cite{Ewald})}
\label{def_met}
Consider a set $T$,  $d:T \times T \rightarrow \mathbb{R} \cup \{ + \infty\}$ is a metric (or a distance function) if the following three conditions are satisfied $\forall x, y, z \in T$:
\begin{enumerate}
\item $d(x,y)=d(y,x)$,
\item $d(x,y)=0 \leftrightarrow x=y$, 
\item $d(x,y) \leq d(x,z)+d(y,z)$.
\end{enumerate}
The ordered pair $(T,d)$ is said to be a metric space.
\end{definition}

\begin{definition}{(Hausdorff Distance \cite{Ewald})}
\label{def_haus}
Assume $X$ and $Y$ are two non-empty subsets of a metric space $(T,d)$. The Hausdorff distance $d_{H}(X,Y)$ is given by:$$\max \{ \underset{x \in X}{\sup} \underset{y \in Y}{\inf} d(x,y), \underset{y \in Y}{\sup} \underset{x \in X}{\inf} d(x,y) \}.$$
\end{definition} 

We also employ the framework from \cite{Abs} to establish relations between different systems. Some relevant notions from that framework
are summarized in the following:
 
\begin{definition}{(System \cite{Abs})}
\label{def_sys}
A system is a sextuple $(X,X_0,U,\longrightarrow,Y,H)$ consisting of:
\begin{itemize}
\item a set of states $X$;
\item a set of initial states $X_0 \subseteq X$;
\item a set of inputs $U$;
\item a transition relation $\longrightarrow \subseteq X \times U \times X$;
\item a set of outputs $Y$;
\item an output map $H: X \rightarrow Y$.
\end{itemize}
\end{definition}

The term finite-state (or infinite-state) system indicates $X$ is a finite (or infinite) set. We employ the shorthand $S=(X,U,\longrightarrow)$ to denote a system when $X=X_0=Y$ and $H: X \rightarrow X$ is the identity map.

\begin{definition}{(Metric System \cite{Abs})}
\label{def_ms}
A system $\mathcal{S}$ is said to be a metric system
if the set of outputs $Y$ is equipped with a metric $d: Y \times Y \rightarrow \mathbb{R}_0^{+}$.
\end{definition}

\begin{definition}{(Approximate Simulation Relation \cite{Abs})}
\label{def_asr}
Consider two metric systems $\mathcal{S}_a$ and $\mathcal{S}_b$ with $Y_a=Y_b$, and let $\epsilon \in \mathbb{R}^+_0$. A relation $R \subseteq X_a \times X_b$
is an $\epsilon$-approximate simulation relation from $\mathcal{S}_a$ to $\mathcal{S}_b$ if the following three
conditions are satisfied:
\begin{enumerate}
\item $\forall x_{a0} \in X_{a0}, \exists x_{b0} \in X_{b0} \text{ such that } (x_{a0},x_{b0}) \in R$;
\item $\forall (x_a,x_b) \in R \text{ we have } d(H_a(x_a),H_b(x_b)) \leq \epsilon$; 
\item $\forall (x_a,x_b) \in R \text{ such that } (x_a,u_a,x'_a) \in \underset{a}{\longrightarrow}$ in $\mathcal{S}_a$ implies $\exists (x_b,u_b,x'_b) \in \underset{b}{\longrightarrow}$ in $\mathcal{S}_b$ satisfying $(x'_a,x'_b) \in R$.
\end{enumerate}
\end{definition}

We denote the existence of an $\epsilon$-approximate simulation relation from $\mathcal{S}_a$ to $\mathcal{S}_b$ by $\mathcal{S}_a\preceq^{\epsilon}\mathcal{S}_b$, and say that $\mathcal{S}_b$ $\epsilon$-approximately simulates $\mathcal{S}_a$. Whenever $\epsilon = 0$, the inequality $d(H_a(x_a), H_b(x_b))\leq \epsilon$ implies $H_a(x_a) = H_b(x_b)$ and the resulting relation is called \emph{(exact) simulation relation}.

Finally, we propose an alternative notion of \emph{quotient system} (see e.g. \cite{Abs} for the traditional definition) to suit our needs:

\begin{definition}{(Power Quotient System)}
\label{def_quo}
Let $\mathcal{S}=(X,X_0,U,\longrightarrow,Y,H)$ be a system and $R$ be an equivalence relation on $X$.
The power quotient of $\mathcal{S}$ by $R$, denoted by $\mathcal{S}_{/R}$, is the system $(X_{/R},X_{/R0},U_{/R},\underset{/R}{\longrightarrow},Y_{/R},H_{/R})$ consisting of:
\begin{itemize}
\item $X_{/R}=X/R$;
\item $X_{/R0}=\{ x_{/R} \in X_{/R} | x_{/R} \cap X_0 \neq \varnothing \}$;
\item $U_{/R}=U$;
\item $(x_{/R},u,x'_{/R}) \in \underset{/R}{\rightarrow}$ if $\exists (x,u,x') \in \rightarrow$ in $\mathcal{S}$ with $x \in x_{/R}$ and $x' \in x'_{/R}$;
\item $Y_{/R} \subset 2^Y$;
\item $H_{/R}(x_{/R})=\underset{x \in x_{/R}}{\cup}H(x)$.
\end{itemize}
\end{definition}

In the following lemma, we provide a result establishing the relationship between a power quotient system and the original system. 

\begin{lemma}
\label{lem_app_sim}
Let $\mathcal{S}$ be a metric system, $R$ be an equivalence relation on $X$, and let the metric system $\mathcal{S}_{/R}$ be the power quotient system of $\mathcal{S}$ by $R$. For any $\epsilon\geq\max_{x \in x_{/R},x_{/R}\in X/R} d(H(x),H_{/R}(x_{/R}))$, with $d$ the Hausdorff distance over the set $2^Y$, $\mathcal{S}_{/R}$ $\epsilon$-approximately simulates $\mathcal{S}$, i.e. 
$\mathcal{S}\preceq^\epsilon \mathcal{S}_{/R}.$
\end{lemma}

\begin{proof}
Let us consider the candidate simulation relation: $R_s\subset X\times X_{/R}$, where $(x, x_{/R})\in R_s$ if and only if $x\in x_{/R}$. From Definition \ref{def_quo}, the three conditions in Definition \ref{def_asr} immediately follow.
\end{proof}

Note that we are using the fact that for a given set $Y$, $Y\subset 2^Y$ in order to use the Hausdorff distance as a common metric between the output sets of the power quotient system and the concrete system.

In Section~\ref{ssec:tr} we also employ the notion of \emph{flow pipe}~\cite{Chut} to discuss reachable sets:

\begin{definition} (Flow Pipe)
The set of reachable states or the flow pipe from an initial set $X_{0,s}$, in the time interval $[\munderbar{\tau}_s,\bar{\tau}_s]$ is denoted by:
\begin{equation}
\begin{array}{l}
\mathcal{X}_{[\munderbar{\tau}_s,\bar{\tau}_s]} (X_{0,s}) = \underset{t\in [\munderbar{\tau}_s,\bar{\tau}_s]}{\bigcup} \mathcal{X}_t (X_{0,s})
\end{array}
\label{my_etc43}
\end{equation}
where $\mathcal{X}_t (X_{0,s})$ denotes the reachable set at time $t$ from $X_{0,s}$, given by:
\begin{equation}
\begin{array}{l}
\mathcal{X}_t (X_{0,s}) = \{ \xi_{x_0}(t)\,|\, x_0 \in X_{0,s} \}.
\end{array}
\label{my_etc42}
\end{equation}
\end{definition}

Finally, in Section~\ref{ssec:tar} we discuss the embedding of the constructed abstractions in the form of \emph{timed safety automata}~\cite{Henz94} (TSA). We revisit that notion following, for compactness, the presentation in~\cite{Bengtsson04}. For details on timed (safety) automata we refer the interested reader to the original works~\cite{Alu, Henz94}.

Assume a finite alphabet $\Sigma$ (actions), and let $\mathcal{C}$ be a set of finitely many real-valued variables (clocks). Consider $\sim \in \{>,\geq, <, \leq \}$, a clock constraint $\delta$ is a conjunctive formula of atomic constraints $c_1\sim k \mbox{~or~} c_1-c_2 \sim k$ for $c_1,c_2 \in \mathcal{C}$, and $k \in \mathbb{N}$. We denote by $\mathcal{B}(\mathcal{C})$ the set of clock constraints.
A clock valuation (or assignment) is a mapping of the form $\mathcal{C} \longrightarrow \mathbb{R}_0^{+}$. Let $u\in(\mathbb{R}_0^{+})^\mathcal{C}$ be a clock valuation and $g \in \mathcal{B}(\mathcal{C})$, then $u \models g$ ($u$ satisfies $g$) iff $g$ evaluates to true using the values from $u$. For $d \in \mathbb{R}^+_0$, $(u + d) (c):=u(c) + d$, $\forall c \in \mathcal{C}$. For $\textbf{c} \subseteq  \mathcal{C}$, $[\textbf{c}\to \textbf{0}]u$ is the clock assignment that maps all clocks in $\textbf{c}$ to $\textbf{0}$ (a vector with all entries equal to zero of the same size of $\textbf{c}$) and agrees with $u$ for the remaining clocks in $\mathcal{C} \setminus \textbf{c}$.

\begin{definition}{(Timed Safety Automata~\cite{Bengtsson04})}
\label{def_TA}
A \emph{timed safety automata} is a tuple $\mathcal{A}=(L,L_0,\Sigma, \mathcal{C}, E,I)$\footnote{Although technically not necessary, for clarity we extend the original definition to explicitly state the actions' and clocks' sets.} where
\begin{itemize}
\item $L$ is a finite set of locations (or discrete states);
\item $L_0 \subseteq L$ is a set of start locations\footnote{Note that for later convenience, and without any impact on the expressivity of the model, we slightly modify the definition of~\cite{Bengtsson04} to allow several possible initial locations as in the original works of~\cite{Alu, Henz94}.};
\item $\Sigma$ is the set of actions;
\item $\mathcal{C}$ is the set of clocks;
\item $E \subseteq L \times \mathcal{B}(\mathcal{C})\times \Sigma \times 2^\mathcal{C}\times L$ is the set of transitions. 
\item $I:L \longrightarrow \mathcal{B}(\mathcal{C})$ assigns invariants to locations.
\end{itemize}
\end{definition}
We use the shorthand notation $l \rTo^{g,a,r}l'$ when $(l,g,a,r,l')\in E$, representing a transition from state $l$ to state $l'$ on input symbol $a$. The set $r\subseteq \mathcal{C}$ indicates the clocks to be reset with this transition, and $g$ is a clock constraint over $\mathcal{C}$ (a guard condition).
\begin{definition}{(Operational Semantics~\cite{Bengtsson04})}
The semantics of a timed safety automaton is a transition system (also known as timed transition system) where states are pairs $(l,u)$, with $l\in L$ and $u$ a clock valuation, and transitions are defined by the rules:
\begin{itemize}
\item $(l,u)\rTo^d(l,u+d)$ if $u\models I(l)$ and $(u+d)\in I(l)$ for a nonnegative real $d\in\mathbb{R}^+$;
\item $(l,u)\rTo^a (l',u')$ if $l\rTo^{g,a,r}l'$, $u\models g$, $u'=[r\to\mathbf{0}]u$ and $u'\models I(l')$.
\end{itemize}
\end{definition}
It is worth remarking that the \emph{guards} of a command assert necessary conditions for the transitions to take place, while \emph{invariants} assert sufficient conditions for transitions to take place, and must not be violated by letting time advance. Thus, invariants establish upper bounds for the time to the next transition (to a different location)~\cite{Henz94}.

\subsection{LTI Event-Triggered Control System}
Consider linear time invariant (LTI) systems without disturbances given by:
\begin{equation}
\begin{array}{l}
\dot{\xi}(t)=A \xi(t)+B\upsilon(t), \; \xi(t) \in \mathbb{R}^{n}, \upsilon(t) \in \mathbb{R}^{m}
\end{array}
\label{my_etc1}
\end{equation}
with linear state-feedback laws implemented in a sample-and-hold manner:
\begin{equation}
\begin{array}{l}
\upsilon(t)=\upsilon(t_k)=K \xi(t_k), \; \forall t \in [t_k,t_{k+1}), \; k \in \mathbb{N}.
\end{array}
\label{my_etc6}
\end{equation}

Let us denote by $\xi_x(t) \in \mathbb{R}^n$ the solution of \eqref{my_etc1}-\eqref{my_etc6} for $t\in[t_k,t_{k+1}]$ with $\xi_x(t_k)=x$ as its initial condition. Furthermore, let us denote by $e_x(t):=x-\xi_x(t)$ the virtual error introduced by the sampling mechanism. 
We consider the event triggering approach proposed in \cite{Tab_etc} which enforces sampling instants given by conditions of the form: 
\begin{equation}
\label{tabtrig}
t_{k+1}=\min \{t> t_k| \; | e_x(t) |^{2}\geq \alpha | \xi_x(t) |^2 \},
\end{equation}
for some appropriate design parameter $\alpha\in\mathbb{R}^+$.

Let the sampling period associated to a state $\xi(t_k)=x$ be denoted by $\tau(x):=t_{k+1}-t_k$. Then, the values of $\xi_x$ and $e_x$ can be expressed in terms of $x$, for $\sigma \in [0, t_{k+1}-t_k]$ as follows: 
\begin{equation}
\begin{array}{l}
\xi_x(t_k+\sigma)=\Lambda (\sigma)x, 
\end{array} 
\label{my_etc15}
\end{equation}
\begin{equation}
\begin{array}{l}
e_x(t_k+\sigma)=[I-\Lambda (\sigma)] x
\end{array} 
\label{my_etc17}
\end{equation}
where
\begin{equation}
\begin{array}{l}
\Lambda (\sigma) = [I+\int_0^{\sigma}e^{Ar}dr(A+BK)].
\end{array} 
\label{my_etc16}
\end{equation}
Thus, substituting \eqref{my_etc15} and \eqref{my_etc17} in \eqref{tabtrig},
the following expression for the state-dependent sampling period is obtained:
\begin{equation}
\begin{array}{l}
\tau(x)=\min \{ \sigma>0 | \; x^T \Phi(\sigma) x=0 \},
\end{array} 
\label{my_etc18_1}
\end{equation}
where
\begin{equation}
\begin{array}{l}
\Phi(\sigma) =[I-\Lambda^{T} (\sigma)] [I-\Lambda (\sigma)]- \alpha \Lambda^{T} (\sigma) \Lambda (\sigma).
\end{array} 
\label{my_etc18}
\end{equation}

\subsection{Problem Statement}
Consider the system: 
\begin{equation*}
\begin{array}{l}
\mathcal{S}=(X,X_0,U,\longrightarrow,Y,H)
\end{array}
\end{equation*}
where 
\begin{itemize}
\item $X=\mathbb{R}^n$;
\item $X_0\subseteq\mathbb{R}^n$;
\item $U=\varnothing$, i.e.\ the system is autonomous;
\item $\longrightarrow \in X \times U \times X$ such that $\forall x, x' \in X: (x,x') \in \longrightarrow$ iff $\xi_x(\tau(x))=x'$;
\item $Y \subset \mathbb{R}^+$;
\item $H: \mathbb{R}^n \rightarrow \mathbb{R}^{+}$ where $H(x)=\tau(x)$.
\end{itemize}
The system $\mathcal{S}$ generates as output sequences all possible sequences of inter-sampling intervals that the system \eqref{my_etc1}-\eqref{my_etc6} with triggering condition~(\ref{tabtrig}) can exhibit. Note in particular that $\mathcal{S}$ is an infinite-state system. 

\begin{problem}
We seek to construct power quotient systems $\mathcal{S}_{/\mathcal{P}}$ based on adequately designed equivalence relations $\mathcal{P}$ defined over the state set $X$ of $\mathcal{S}$.
\end{problem}

In particular, we propose to construct the system $\mathcal{S}_{/\mathcal{P}}$ as follows:
\begin{equation*}
\begin{array}{l}
\mathcal{S}_{/\mathcal{P}}=(X_{/\mathcal{P}},X_{0/\mathcal{P}},U_{/\mathcal{P}}, \underset{/\mathcal{P}}{\longrightarrow},Y_{/\mathcal{P}},H_{/\mathcal{P}})
\end{array}
\end{equation*}
where 
\begin{itemize}
\item $X_{/\mathcal{P}}= \mathbb{R}_{/\mathcal{P}}^n:= \{\mathcal{R}_1,\dots,\mathcal{R}_q \}$;
\item $X_{0/\mathcal{P}}=\{\mathcal{R}_i\,|\, X_0\cap\mathcal{R}_i\neq\emptyset\}$;
\item $U_{/\mathcal{P}}=\varnothing$, i.e.\ the system is autonomous;
\item $(x_{/\mathcal{P}},x'_{/\mathcal{P}}) \in
\underset{/\mathcal{P}}{\longrightarrow} $ if $\exists x \in x_{/\mathcal{P}}$, $\exists x' \in x'_{/\mathcal{P}}$ such that $\xi_x(H(x))=x'$;
\item $Y_{/\mathcal{P}} \subset 2^Y \subset \mathbb{IR}^{+}$;
\item \small{$H_{/\mathcal{P}}(x_{/\mathcal{P}})= [\underset{x \in x_{/\mathcal{P}}}{\min}H(x),\underset{x \in x_{/\mathcal{P}}}{\max}H(x)]:=[\munderbar{\tau}_{x_{/\mathcal{P}}},\bar{\tau}_{x_{/\mathcal{P}}}]$}.
\end{itemize}

The remaining question is now how to: select an appropriate equivalence relation $\mathcal{P}$, compute the respective intervals $[\munderbar{\tau}_{x_{/\mathcal{P}}},\bar{\tau}_{x_{/\mathcal{P}}}]$, and determine when there is a transition between a pair of abstract states $(x_{/\mathcal{P}},x'_{/\mathcal{P}})$. These three questions, and the respective constructions, are addressed in the following section.

\section{Construction of the Abstraction}
\label{sec_4}

In this section, we address the construction of $\mathcal{S}_{/\mathcal{P}}$, and more specifically $X_{/\mathcal{P}}$, $H_{/\mathcal{P}}$, and ${\underset{/\mathcal{P}} {\longrightarrow}}$. The set $X_{/\mathcal{P}}$ and map $H_{/\mathcal{P}}$ are constructed employing a two-step approach inspired by the work in \cite{Fit}. Nonetheless, some modifications are introduced to make that approach suitable to our goals. In \cite{Fit} the minimum inter-sample times, $\munderbar{\tau}_s$, are computed for conic regions and Lyapunov-Razumikhin stability conditions. We adapt that approach (in our Lemma \ref{lem3} and Theorem \ref{them3}) to apply instead to the ISS based triggering mechanism in \eqref{my_etc18_1}-\eqref{my_etc18}. More importantly, we modify the approach to also provide upper limits,  $\bar{\tau}_s$, on the regional inter-sample times (see our Lemma \ref{lemm_u1} and Theorem \ref{them4}). 
We also propose, inspired by the work on stability analysis of switched systems using multiple Lyapunov functions~\cite{Deca}, a modification to the way the S-procedure is employed in~\cite{Fit} for $n$-dimensional state spaces with $n \geq 3$ (see Theorems \ref{them3} and \ref{them4}). 
Finally, we suggest to use a reachability analysis like the one in~\cite{Chut} in order to construct the transition relation ${\underset{/\mathcal{P}} {\longrightarrow}}$. 

\subsection{State Set}
The following observation is the cornerstone in removing spatial dependencies from (\ref{my_etc18_1}) and mapping the infinite-state system $\mathcal{S}$ to a finite-state system $\mathcal{S}_{/\mathcal{P}}$.

\begin{remark}\label{rem_1}
Excluding the origin, all the states which lie on a line that goes through the origin have the same inter-sample time, i.e., $\tau(x)=\tau(\lambda x)$, $\forall \lambda \neq 0$ \cite{Fit}.
\end{remark}

Based on Remark \ref{rem_1} and the fact that a convex polyhedral cone (pointed at the origin) is the union of an infinite number of rays, a convenient approach to partition the state space is via abstracting it into a finite number of convex polyhedral cones (pointed at the origin) $\mathcal{R}_s$ where $s \in \{1,\dots,q \}$ and $\bigcup_{s=1}^q \mathcal{R}_s = \mathbb{R}^n$ (see Figure \ref{abs_graph}). 

The state space abstraction technique proposed by \cite{Fit}, called \emph{isotropic covering}, is briefly explained in the following. Generalized spherical coordinates $x \in \mathbb{R}^n : \; (r, \theta_1, \dots,\theta_{n-1})$ are used for the abstraction purpose, where $r=|x|$ and $\theta_1, \dots,\theta_{n-1}$ are the corresponding angular coordinates of $x$ and $\theta_1, \dots,\theta_{n-2} \in [0,\pi]$  and $\theta_{n-1} \in [-\pi,\pi]$. Each angular coordinate is divided into $\bar{m}$ (not necessarily equidistant) intervals. Hence, the number of conic regions $q$ is equal to $\bar{m}^{(n-1)}$ (see \cite{Fit} for more details on constructing these state space partitions). Remark \ref{rem_1} also suggests that it suffices to only consider half of the state space since $x$ and $-x$ behave in the same way in (\ref{my_etc18_1}). Therefore, one can consider half of the state space (for example, by assuming $\theta_{n-1} \in [0,\pi]$), and then appropriately map the results to the other half of the state space. In this case, $q$ is equal to $2 \times \bar{m}^{(n-1)}$.
The resulting polyhedral cones admit the following representations: 
\begin{eqnarray}
\label{eq:poly_rep}
\notag
\mathcal{R}_s=&\{ x \in \mathbb{R}^2| \; x^T Q_s x \geq 0 \},&\;\text{if}\, n=2\\ 
\mathcal{R}_s=&\{ x \in \mathbb{R}^n| \; E_s x \geq 0\},&\; \text{if}\, n\geq 3 
\end{eqnarray} 
for $s\in\{1,\ldots,\,q\}$ and some appropriately designed matrices $Q_s=Q_s^T \in \mathcal{M}_2(\mathbb{R})$ or $E_s \in \mathcal{M}_{n\times p}(\mathbb{R})$ with $p\leq 2n-2$.

\subsection{Output Map}
We focus now on the construction of the output map $H_{/\mathcal{P}}$ (and the output set $Y_{/\mathcal{P}}$ in the process). We propose a constructive method to find a time interval $[\munderbar{\tau}_s,\bar{\tau}_s]$ associated to each region $\mathcal{R}_s$ such that $\forall x \in \mathcal{R}_s: \tau(x) \in [\munderbar{\tau}_s,\bar{\tau}_s]$, with $\tau(\cdot)$ as in (\ref{my_etc18_1}). We construct convex polytopes around the matrix $\Phi(\sigma)$ (see~\cite{Het}), resulting in LMIs which combined with the S-procedure allows us to find bounds $\munderbar{\tau}_s$ and $\bar{\tau}_s$ on the different regions as specified in the previous subsection.

We illustrate first how to derive $\munderbar{\tau}_s$ so that in the time interval $[0,\munderbar{\tau}_s]$ no triggering is enabled. The main idea of the following lemma is to construct a finite set of matrices $\munderbar\Phi_{\kappa,s}$, with $\kappa \in \mathcal{K}_s$ (a finite set of indices), such that: 
$$(x^T \munderbar\Phi_{\kappa,s} x \leq 0, \forall \kappa \in \mathcal{K}_s)\Longrightarrow (x^T \Phi(\sigma) x \leq 0, \forall \sigma \in [0,\munderbar{\tau}_s]).$$

\begin{figure}[t]
 \centering
  \includegraphics[width=0.6
  \linewidth]{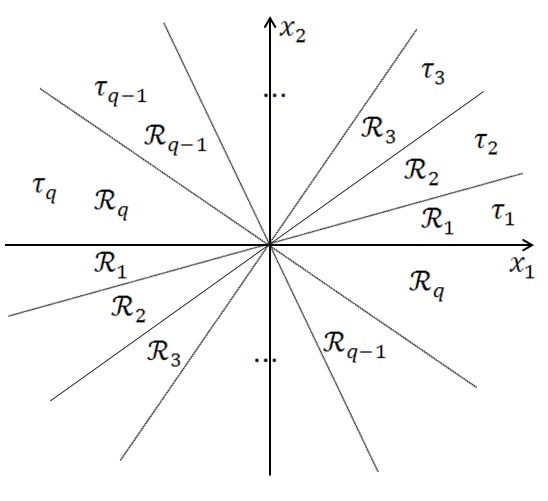}
 \caption{An illustrative $2$-dimensional example of the state space abstraction using convex polyhedral cones.} 
 \label{abs_graph}
\end{figure}

Let the scalar $\bar{\sigma}>0$ denote a time instant for which the triggering mechanism (\ref{my_etc18_1}) is enabled in the whole state space, i.e., $x^T \Phi(\bar\sigma)x \geq 0,\forall x \in \mathbb{R}^n$ (Remark \ref{rem_sig} addresses the selection of $\bar{\sigma}$). Select an integer $N_{conv}\geq 0$ such that $N_{conv}+1$ is the number of vertices considered for the polytope containing $\Phi(\sigma)$, and an integer $l\geq1$ determining the number of time subdivisions considered in the time interval $[0,\bar{\sigma}]$.

\begin{lemma}
\label{lem3}
Consider a time limit $\munderbar{\tau}_s \in (0,\bar{\sigma}]$. If 
$$x^T \munderbar{\Phi}_{(i,j),s} x \leq 0,$$ 
$$\forall(i,j) \in \mathcal{K}_s=(\{0,\dots,N_{conv} \} \times \{0,\dots,\lfloor \frac{\munderbar{\tau}_s l}{\bar{\sigma}} \rfloor \}),$$
then 
$$x^T \Phi(\sigma) x \leq 0\,,\forall \sigma \in [0,\munderbar{\tau}_s],$$ 
with $\Phi$ as in (\ref{my_etc18}) and
\begin{equation*}
\begin{array}{l}
\munderbar{\Phi}_{(i,j),s}:=\munderbar{\hat{\Phi}}_{(i,j),s}+ \munderbar{\nu} I, 
\end{array}
\end{equation*}

\begin{equation}
\begin{array}{l}
\munderbar{\hat{\Phi}}_{(i,j),s}:=\left\{\begin{array}{lc}
\sum_{k=0}^i \; L_{k,j} (\frac{\bar{\sigma}}{l})^k  \; \; \; if \; j < \lfloor \frac{\munderbar{\tau}_s l}{\bar{\sigma}} \rfloor, \\
\sum_{k=0}^i \; L_{k,j} (\munderbar{\tau}_s -\frac{j \bar{\sigma}}{l})^k  \; \; \; j=\lfloor\frac{\munderbar{\tau}_s l}{\bar{\sigma}} \rfloor,
\end{array}\right. 
\end{array}
\label{my_etc20}
\end{equation}

\begin{equation}
\begin{array}{l}
\left\{\begin{array}{ll}
L_{0,j}:=& I-\Pi_{1,j}-\Pi^{T}_{1,j}+(1-\alpha) \Pi^{T}_{1,j} \Pi_{1,j}, \\
L_{1,j}:=& [(1-\alpha)\Pi^{T}_{1,j}-I] \Pi_{2,j}\\  & +\Pi^{T}_{2,j}[(1-\alpha)\Pi_{1,j}-I], \\
L_{k\geq 2,j}:=& [(1-\alpha) \Pi^{T}_{1,j}-I] \frac{A^{k-1}}{k!} \Pi_{2,j}\\
 & +\Pi^{T}_{2,j} \frac{(A^{k-1})^{T}}{k!} [(1-\alpha) \Pi_{1,j}-I]\\
 & +(1-\alpha) \Pi^{T}_{2,j}(\sum^{k-1}_{i=1} \frac{(A^{i-1})^{T}}{i!} \frac{A^{k-i-1}}{(k-i)!}) \Pi_{2,j}, 
\end{array}\right. 
\end{array}
\label{my_etc21}
\end{equation}
\begin{equation}
\begin{array}{l}
\left\{\begin{array}{l}
\Pi_{1,j}:=I+M_j (A+BK), \\
\Pi_{2,j}:=N_j(A+BK),
\end{array}\right. 
\end{array}
\label{my_etc22}
\end{equation}
\begin{equation}
\begin{array}{l}
M_j:= \int^{j \frac{\bar{\sigma}}{l}}_0 e^{As}ds, \; \; \; N_j:=AM_j+I,
\end{array}
\label{my_etc23}
\end{equation}
and
\begin{equation}
\begin{array}{l}
\munderbar{\nu} \geq \underset{\sigma' \in [0,\frac{\bar{\sigma}}{l}],\; r \in \{ 0,\dots, l-1\}}{\max} \lambda_{\max} (\Phi(\sigma'+r\frac{\bar{\sigma}}{l})-\tilde{\Phi}_{N_{conv},r}(\sigma')), \\
\end{array}
\label{my_etc24}
\end{equation}
where
\begin{equation}
\begin{array}{l}
\tilde{\Phi}_{N_{conv},r}(\sigma'):= \sum_{k=0}^{N_{conv}} \; L_{k,r} \sigma'^{k}. 
\end{array}
\label{my_etc25}
\end{equation}
\end{lemma}
\begin{proof}
See Appendix.
\end{proof}

Then leveraging the S-procedure the following theorem provides an approach to regionally reduce the conservatism involved in $\munderbar{\tau}_s$ estimates obtained from Lemma~\ref{lem3}. 

\begin{theorem}[Regional Lower Bound Approximation]
\label{them3}
Consider the inter-sampling time set $\{\munderbar{\tau}_1, \ldots, \munderbar{\tau}_q\}$ and matrices $\munderbar{\Phi}_{\kappa,s}$ satisfying $\forall s \in \{ 1, \ldots, q \}$, $\forall\, \kappa=(i,j) \in \mathcal{K}_s$, 
$0<\munderbar{\tau}_s \leq \bar{\sigma}$, $\munderbar{\Phi}_{\kappa,s}  \preceq 0$. If there exist scalars $\munderbar{\varepsilon}_{\kappa,s}\geq 0$ (for $n=2$) or symmetric matrices $\munderbar{U}_{\kappa,s}$ with nonnegative entries (for $n \geq 3$) such that the LMIs 
\begin{equation}
\begin{array}{l}
\munderbar{\Phi}_{\kappa,s}+\munderbar{\varepsilon}_{\kappa,s} Q_s \preceq 0 \; \; \mbox{~if~} \; \; n=2
\end{array}
\label{my_etc32}
\end{equation}
\begin{equation}
\begin{array}{l}
\munderbar{\Phi}_{\kappa,s}+E_s^{T} \munderbar{U}_{\kappa,s} E_s \preceq 0  \; \; \mbox{~if~} \; \; n\geq 3
\end{array}
\label{my_etc33}
\end{equation}
hold, then $\forall x \in \mathcal{R}_s$ determined by (\ref{eq:poly_rep}) the inter-sample time (\ref{tabtrig}) of the control system (\ref{my_etc1}-\ref{my_etc6}) is regionally bounded from below by $\munderbar{\tau}_s$.
\end{theorem} 

\begin{proof}
The result is a direct consequence of applying a lossless (if $n=2$) or a lossy (if $n \geq 3$) version of the S-procedure:
if the stated conditions hold then, for every $x\in\mathbb{R}^n$, $x^TQ_sx\geq 0$ ($x^TE_s^{T} \munderbar{U}_{\kappa,s} E_sx\geq 0$) implies that $x^{T} \munderbar{\Phi}_{\kappa,s}x\leq 0$.
From~(\ref{eq:poly_rep}) we have that $x^TQ_sx\geq 0$ ($E_sx\geq 0$) iff $x\in\mathcal{R}_s$. Thus one can conclude that for every $x\in\mathcal{R}_s$, $x^{T} \munderbar{\Phi}_{\kappa,s}x\leq 0$ and by Lemma~\ref{lem3} the result follows.
\end{proof}

A similar approach can be employed to compute upper bounds $\bar{\tau}_s$ of the regional inter-sample time.
This is summarized in Lemma \ref{lemm_u1} and Theorem \ref{them4}.

\begin{lemma}
\label{lemm_u1}
Consider a time limit $\bar{\tau}_s \in (\munderbar{\tau}_s,\bar{\sigma}]$. If
$$x^T \bar{\Phi}_{(i,j),s} x \geq 0,$$ 
$$\forall (i,j) \in \mathcal{K}_s=(\{0,\dots,N_{conv} \} \times \{ \lfloor \frac{\bar{\tau}_s l}{\bar{\sigma}} \rfloor ,\dots, l-1 \}),$$
then 
$$x^T \Phi(\sigma) x \geq 0,\, \forall \sigma \in [\bar{\tau}_s,\bar{\sigma}],$$ with 
$\Phi$ as in (\ref{my_etc18}) and
\begin{equation*}
\begin{array}{l}
\bar{\Phi}_{(i,j),s}:=\bar{\hat{\Phi}}_{(i,j),s}+ \bar{\nu} I, 
\end{array}
\end{equation*}

\begin{equation}
\begin{array}{l}
\bar{\hat{\Phi}}_{(i,j),s}:=\left\{\begin{array}{lc}
\sum_{k=0}^i \; L_{k,j} ((j+1)\frac{ \bar{\sigma}}{l}-\bar{\tau}_s )^k \; if  \; j = \lfloor \frac{\bar{\tau}_s l}{\bar{\sigma}} \rfloor, \\
\sum_{k=0}^i \; L_{k,j} (\frac{\bar{\sigma}}{l})^k  \; \; \; if \; j > \lfloor \frac{\bar{\tau}_s l}{\bar{\sigma}} \rfloor,
\end{array}\right. 
\end{array}
\label{my_etc20_1}
\end{equation}
\begin{equation}
\begin{array}{l}
\bar{\nu} \leq \underset{\sigma' \in [0,\frac{\bar{\sigma}}{l}],\; r \in \{ 0,\dots, l-1\}}{\max} \lambda_{\min} (\Phi(\sigma'+r\frac{\bar{\sigma}}{l})-\tilde{\Phi}_{N_{conv},r}(\sigma')), \\
\end{array}
\label{my_etc24_1}
\end{equation}
and $L_{k,j}$ given by (\ref{my_etc21}). 
\end{lemma}
\begin{proof}
See appendix. \end{proof}

\begin{theorem}[Regional Upper Bound Approximation]
\label{them4}
Consider the inter-sampling time set $\{ \bar{\tau}_1, \ldots, \bar{\tau}_q \}$ 
and matrices $\bar{\Phi}_{\kappa,s}$ satisfying $\forall s \in \{ 1, \ldots, q \}$, $\forall \, \kappa=(i,j) \in \mathcal{K}_s=(\{0,\dots,N_{conv} \} \times \{\lfloor \frac{\bar{\tau}_s l}{\bar{\sigma}} \rfloor,\dots,l-1 \})$, 
$\munderbar{\tau}_s <\bar{\tau}_s \leq \bar{\sigma}$, $ \bar{\Phi}_{\kappa,s}  \succeq 0$.
If there exist scalars $\bar{\varepsilon}_{\kappa,s}\geq 0$ (for $n=2$) or symmetric matrices $\bar{U}_{\kappa,s}$ with nonnegative entries (for $n \geq 3$) such that the LMIs 
\begin{equation}
\begin{array}{l}
\bar{\Phi}_{\kappa,s}-\bar{\varepsilon}_{\kappa,s} Q_s \succeq 0 \; \; \mbox{~if~} \; \; n=2
\end{array}
\label{my_etc40}
\end{equation}
\begin{equation}
\begin{array}{l}
\bar{\Phi}_{\kappa,s}-E_s^{T} \bar{U}_{\kappa,s} E_s \succeq 0  \; \; \mbox{~if~} \; \; n\geq 3
\end{array}
\label{my_etc41}
\end{equation}
hold, then $\forall x \in \mathcal{R}_s$ determined by (\ref{eq:poly_rep}) the inter-sampling time (\ref{tabtrig}) of the control system (\ref{my_etc1}-\ref{my_etc6}) is regionally bounded from above by $\bar{\tau}_s$.
\end{theorem} 

\begin{proof}
Analogous to the proof of Theorem~\ref{them3}. \end{proof}

A procedural algorithm illustrating the use of Theorems \ref{them3} and \ref{them4} in practice to compute $\munderbar{\tau}_s$ and $\bar{\tau}_s$ is as follows: 
\begin{enumerate}
\item Derive $\munderbar{\nu}$ using (\ref{my_etc24}). Consider $s=1$  in Theorem \ref{them3}. Implement a line search on $\munderbar{\tau}'$ in the interval $[0,\bar{\sigma}]$ to find the lower bound on the inter-sample time in the whole state space. 
\item Partition the state space into $q$ regions and again use Theorem \ref{them3} to find $\{\munderbar{\tau}_1,\dots,\munderbar{\tau}_q \}$. This optimization problem is a line search on $\munderbar{\tau}_s$ in the time interval $[\munderbar{\tau}',\bar{\sigma}]$ combined with LMI feasibility problems on $\munderbar{\varepsilon}_{\kappa,s}$ or $\munderbar{U}_{\kappa,s}$ at each line search iteration. 
\item Derive $\bar{\nu}$ using (\ref{my_etc24_1}) and then use Theorem \ref{them4} to find $\bar{\tau}_s$. This problem is a combination of a line search on $\bar{\tau}_s$ in the time interval $[\munderbar{\tau}_s, \bar{\sigma}]$ with LMI feasibility problems on $\bar{\varepsilon}_{\kappa,s}$ or $\bar{U}_{\kappa,s}$. 
\end{enumerate}

\begin{remark}
\label{rem_sig}
To the best of our knowledge, there is no formal result establishing an upper bound $\bar{\sigma}$ to $\munderbar{\tau}_s$ and $\bar{\tau}_s$. In practice, nonetheless, a line search is sufficient: increasing the value of $\bar{\sigma}$ until the derived values of $\munderbar{\tau}_s$ and $\bar{\tau}_s$ satisfy $\munderbar{\tau}_s, \bar{\tau}_s < \bar{\sigma}$.
\end{remark}

\subsection{Transition Relation}
\label{ssec:tr}
In this subsection, we address the construction of the transition map ${\underset{/\mathcal{P}} {\longrightarrow}}$. To do so, each conic region is considered as a discrete mode of $\mathcal{S}_{/\mathcal{P}}$ and a reachability analysis is catered to derive all possible transitions in $\mathcal{S}_{/\mathcal{P}}$. 


We consider initial sets $X_{0,s}$ in each $\mathcal{R}_s$ which are polytopes with their vertices lying on the extreme rays of $\mathcal{R}_s$ excluding the origin, e.g. for a two dimensional system the initial sets will be line segments connecting the boundary rays of the respective $\mathcal{R}_s$ sets. 
This selection is justified by the following facts: states that lie on a line going through the origin have an identical triggering behavior (see Remark \ref{rem_1}); these states are mapped to another line that goes through the origin (since the state evolution (\ref{my_etc15}) is a linear map on $x$); and the image of a convex polytope under a linear map is another convex polytope.


Given $X_{0,s}$ in  $\mathcal{R}_s$ with corresponding $\munderbar{\tau}_s$ and $\bar{\tau}_s$, we seek first to derive a polytopic outer approximation $\hat{\mathcal{X}}_{[\munderbar{\tau}_s,\bar{\tau}_s]}(X_{0,s})$ to the flow pipe $\mathcal{X}_{[\munderbar{\tau}_s,\bar{\tau}_s]}(X_{0,s})$, i.e. $$\mathcal{X}_{[\munderbar{\tau}_s,\bar{\tau}_s]}(X_{0,s}) \subseteq \hat{\mathcal{X}}_{[\munderbar{\tau}_s,\bar{\tau}_s]}(X_{0,s}).$$ 

Several methods can be catered to solve this problem, in particular we employ in our implementations the approach from~\cite{Chut}.
A generic approach to reduce the conservatism of the approximation, also employed in~\cite{Chut}, is to divide the time interval $[\munderbar{\tau}_s,\bar{\tau}_s]$ into $\bar{f}$ subintervals and compute an over approximation  $\hat{\mathcal{X}}_{[\munderbar{\tau}_s,\bar{\tau}_s]}(X_{0,s})$ as the union of a sequence of convex polytopes $\hat{\mathcal{X}}_{[t_f,t_{f+1}]}(X_{0,s})$, $\forall f \in \{1,\dots,\bar{f} \}$ with $t_1=\munderbar{\tau}_s$ and $t_{\bar{f}+1}=\bar{\tau}_s$:

\begin{equation}
\begin{array}{l}
\hat{\mathcal{X}}_{[\munderbar{\tau}_s,\bar{\tau}_s]}(X_{0,s})= \underset{f}{\cup} \; \hat{\mathcal{X}}_{[t_f,t_{f+1}]}(X_{0,s}).
\end{array}
\label{my_etc50}
\end{equation}

In order to derive the transitions in $\mathcal{S}_{/\mathcal{P}}$, the intersection between the flow pipe corresponding to each region $\mathcal{R}_s$ and the initial conic regions need to be founded. This can be done by solving the following feasibility problem for each pair of conic regions $(\mathcal{R}_s, \mathcal{R}_{s'})$:

\begin{equation}
\begin{array}{ll}
\text{Feas} & C_{f,s} \bar{x} \leq d_{f,s}^{\star}\\
 & E_{s'} \bar{x} \geq 0
\end{array}
\label{my_etc51}
\end{equation}
where $\{\bar{x} \in \mathbb{R}^n| \; C_{f,s} \bar{x} \leq d_{f,s}^{\star} \}$ represents the outer approximation of the $f$-th flow pipe segment associated with $\mathcal{R}_s$, and $\{\bar{x} \in \mathbb{R}^n| \; E_{s'} \bar{x} \geq 0 \}$ is the initial conic region where $s' \in \{ 1, \dots, q \}$. The feasibility of (\ref{my_etc51}) for any value of $f \in \{1,\dots, \bar{f} \}$ indicates that there is a transition from mode $s$ to mode $s'$ in $\mathcal{S}_{/\mathcal{P}}$.

\section{Main Results}
\label{sec_5}

\subsection{$\epsilon$-Approximate Simulation Relation}

We summarize in the following theorem the main result of the paper.
\begin{theorem}
The metric system $$\mathcal{S}_{/\mathcal{P}}=(X_{/\mathcal{P}},X_{0/\mathcal{P}},U_{/\mathcal{P}}, \underset{/\mathcal{P}}{\longrightarrow},Y_{/\mathcal{P}},H_{/\mathcal{P}}),$$  $\varepsilon$-approximately simulates $\mathcal{S}$ where $\epsilon= \max d_H(y,y')$, $y=H(x) \in Y, y'=H_{/\mathcal{P}}(x') \in Y/\mathcal{P}$, $\forall (x,x') \in \mathcal{P}$, and $d_H(\cdot,\cdot)$ is the Hausdorff  distance. 
\end{theorem}
\begin{proof}
This is a direct consequence of Lemma \ref{lem_app_sim} and the construction described in the previous section.
\end{proof}

\begin{remark}
\label{coro_1}
Refining the conic regions of a given abstraction into a lager number of conic regions $\mathcal{R}_s$, $q$, the length of the sampling intervals in each region $|\bar{\tau}_s - \munderbar{\tau}_s|$ cannot increase. Thus, increasing $q$ results in abstractions with a new precision $\epsilon ' \leq \epsilon$ where $\epsilon '$ is computed based on the refined state space abstraction. 
Note that this is not necessarily the case if the partition of the state-space in the new abstraction is not a refinement of the original abstraction.
\end{remark}


\subsection{Timed-Automaton Representation}
\label{ssec:tar}

Finally, we show that $\mathcal{S_{/\mathcal{P}}}$ is in fact equivalent to a TSA. Regardless of TSAs uncountable state space (related to its clock variables), it has been shown that its reachability analysis is decidable \cite{Alu}. This fact makes TSA a powerful tool to model real-world systems, such as real-time systems, where discrete transitions are coupled with timing constraints. Several tools, exploring the attractive features of TSA, have been developed for verification and synthesis \cite{Alur} and \cite{Upp}. 

In the following, the semantics of $\mathcal{S}_{/\mathcal{P}}$ is spelled out. The corresponding output $y_{/\mathcal{P}} \in Y_{/\mathcal{P}}$ of a state $x_{/\mathcal{P}} \in X_{/\mathcal{P}}$ indicates that $\mathcal{S_{/\mathcal{P}}}$: 
\begin{enumerate}
\item remains at $x_{/\mathcal{P}}$ during the time interval $[0,\munderbar{\tau}_{x_{/\mathcal{P}}})$, 
\item possibly leaves $x_{/\mathcal{P}}$ during the time interval $[\munderbar{\tau}_{x_{/\mathcal{P}}},\bar{\tau}_{x_{/\mathcal{P}}})$, and
\item is forced to leave $x_{/\mathcal{P}}$ at $\bar{\tau}_{x_{/\mathcal{P}}}$.
\end{enumerate}


Assume for simplicity that $\mathcal{S}_{/\mathcal{P}}$ has a single initial state. Based on Definition \ref{def_TA} and the previous description, one can equivalently represent the same semantics expressed by $\mathcal{S}_{/\mathcal{P}}$ with the TSA $(L,L_0,\Sigma, \mathcal{C}, E,I)$ where: 
\begin{itemize}
\item the set of locations $L:=X_{/\mathcal{P}}=\{l_1,\dots,l_q \}$;
\item the initial location $L_0 :=X_{0/\mathcal{P}}$; 
\item the set of actions $\Sigma = \{*\}$ is an arbitrary labeling of discrete transitions (or edges);
\item the clock set $\mathcal{C}=\{c\}$ contains a single clock; 
\item the set of edges $E$ is such that $(l_s,g,a,r,l_{s'}) \in E$ iff $l_s \rTo_{/\mathcal{P}} l_{s'}$, $g=\{\munderbar{\tau}_s \leq c \leq \bar{\tau}_s \}$, $a=*$, and $r=\{ c:=0\}$;
\item the invariant map $I(l_s):=\{0\leq c \leq \bar{\tau}_s\}$, $\forall s \in \{1,\dots,q\}$.
\end{itemize}

\section{Numerical Example}
\label{sec_6} 
In this section, we illustrate the results of this study. Consider a linear system, employed in the example in~\cite{Tab_etc}, with state feedback law that is given by:

\begin{equation}
\begin{array}{l}
\dot{\xi}(t) = \left[ \begin{array}{cc}
0 & 1  \\ -2 & 3  \end{array} \right] \xi(t) +  \left[ \begin{array}{c}
0  \\ 1   \end{array} \right] \nu(t),\\
\\
\nu(t) = [1 \;\; -4]\xi(t).
\end{array}
\label{my_etc60}
\end{equation}

We set the triggering coefficient $\alpha = 0.05$, the order of polynomial approximation $N_{conv}=5$, the number of polytopic subdivisions $l=100$, the number of angular subdivisions $\bar{m}=10$ (hence, the number of conic regions is $q= 2 \times \bar{m}^{(n-1)}=2 \times 10^{(2-1)}=20$), the upper bound of the inter-sample interval  $\bar{\sigma}=1\sec$, and the time step employed to subdivide the flow pipe segments is $0.01\sec$.  

Figure \ref{etc_ex1_1} depicts $\munderbar{\tau}_s$ and $\bar{\tau}_s$ for the closed loop system given by (\ref{my_etc60}) using Theorem \ref{them3} and Theorem \ref{them4}, respectively. In Figure \ref{etc_ex1_9}, another representation of  $\munderbar{\tau}_s$ and $\bar{\tau}_s$ is given in the state space. Based on the Hausdorff distance and the derived boundaries, one gets that in this example the precision of the abstraction is $\epsilon=0.119$.

\begin{figure}[h]
 \centering
  \includegraphics[width=1
  \linewidth]{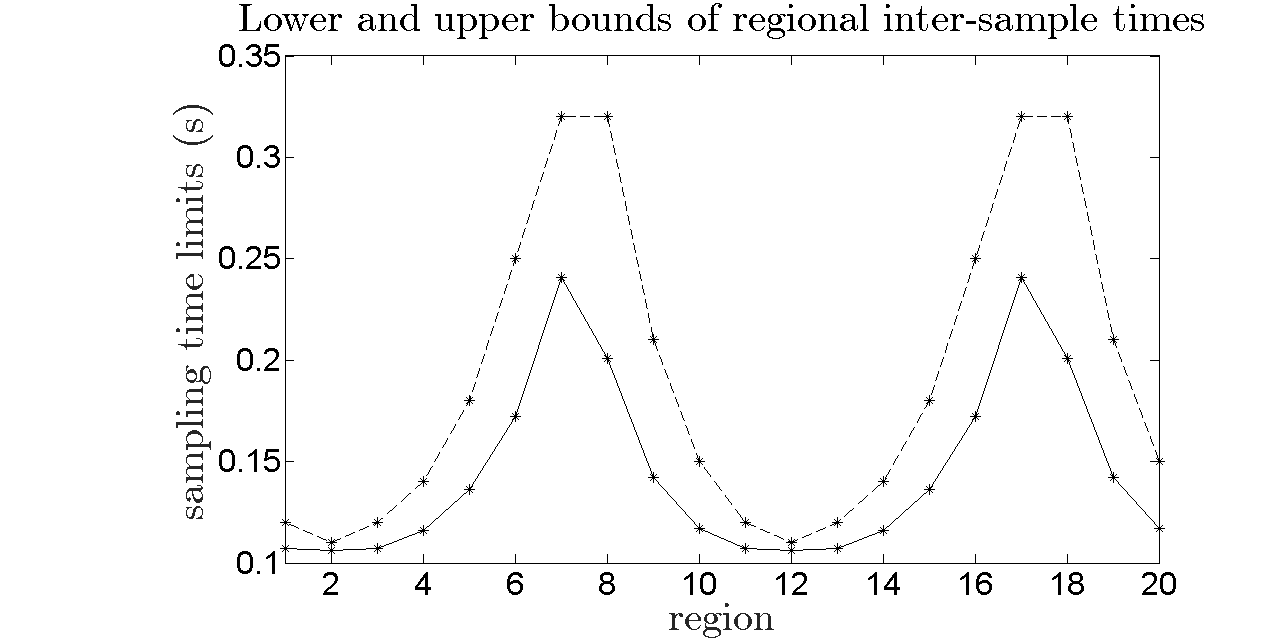}
 \caption{Lower and upper bounds approximation of regional inter-sample times are depicted by solid and dashed curves, respectively.} 
 \label{etc_ex1_1}
\end{figure}

\begin{figure}[h]
 \centering
  \includegraphics[width=1.1
  \linewidth]{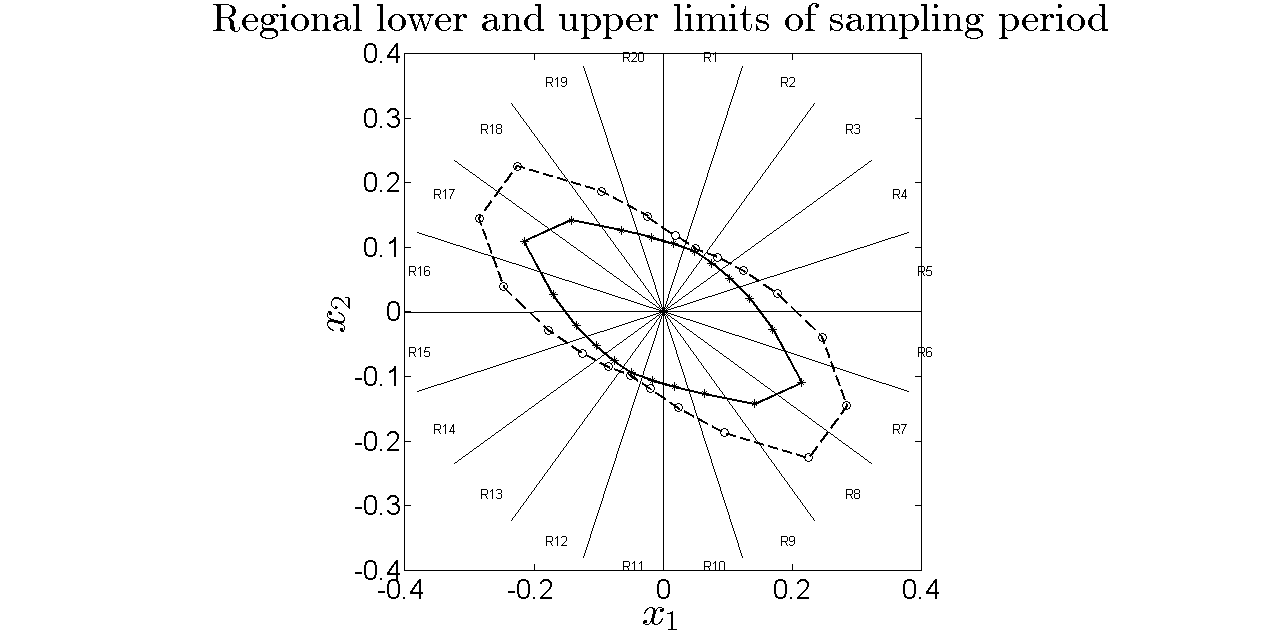}
 \caption{Lower and upper bounds approximation of regional inter-sample times are depicted by solid and dashed lines, respectively. The distance to the origin of each asterisk (circle) denotes the regional value of lower (upper) bound.} 
 \label{etc_ex1_9}
\end{figure}

Note that by increasing $\bar{m}$: $\munderbar{\tau}_s$ and $\bar{\tau}_s$ will possibly shift upward and downward, respectively. By doing so, one can derive tighter bounds for the given event-triggered control system. Thus, a more precise $\epsilon$-approximate simulation relation from $\mathcal{S}$ to $\mathcal{S}_{/\mathcal{P}}$ can be achieved. Figure \ref{etc_ex1_10} shows a comparison between two different state space abstraction with $\bar{m}=10$ and $100$. Evidently, partitioning half of the state space to $100$ regions instead of $10$ regions leads to tighter bounds where $\epsilon '=0.021$.  

\begin{figure}[h]
 \centering
  \includegraphics[width=0.9
  \linewidth]{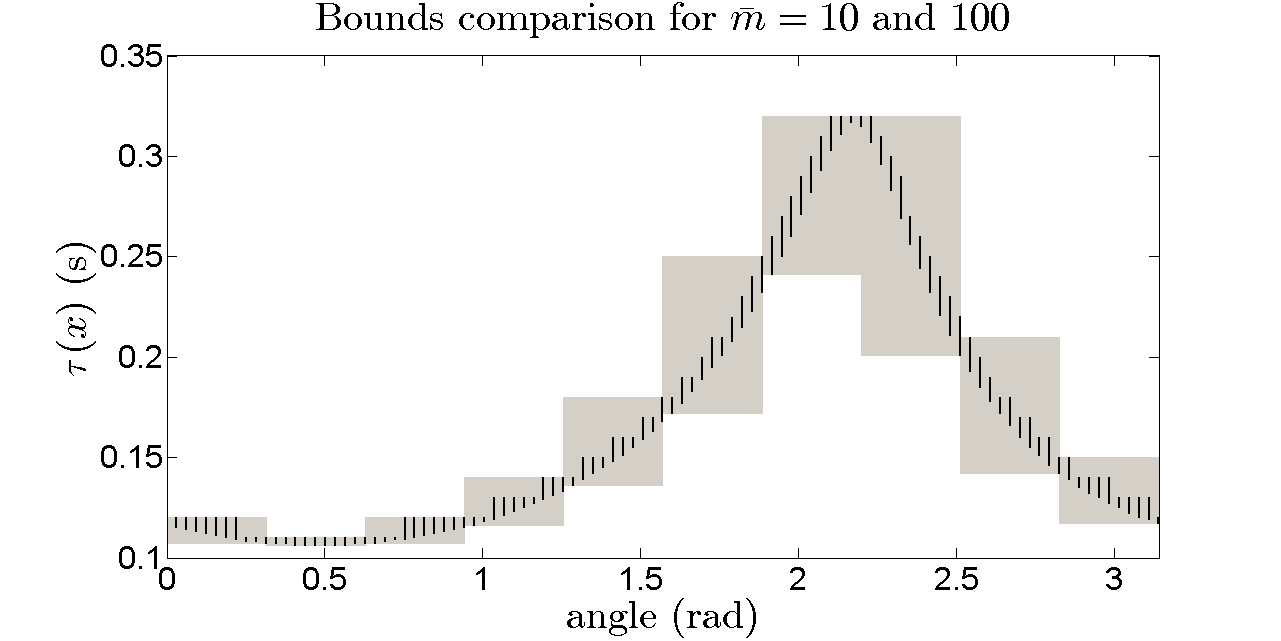}
 \caption{Comparison of two different state space abstractions with $\bar{m}=10$ (the light shaded area) and $\bar{m}=100$ (the dark dashed area) for half of the state space.} 
 \label{etc_ex1_10}
\end{figure}

Figure \ref{etc_ex1_3} illustrates the validity of the theoretical bounds that we found for $\munderbar{\tau}_s$ (black line) and $\bar{\tau}_s$ (dashed line). The asterisks represent the inter-sample times sequence during $5 \sec$ simulation of the event-triggered control system.

\begin{figure}[h]
 \centering
  \includegraphics[width=1
  \linewidth]{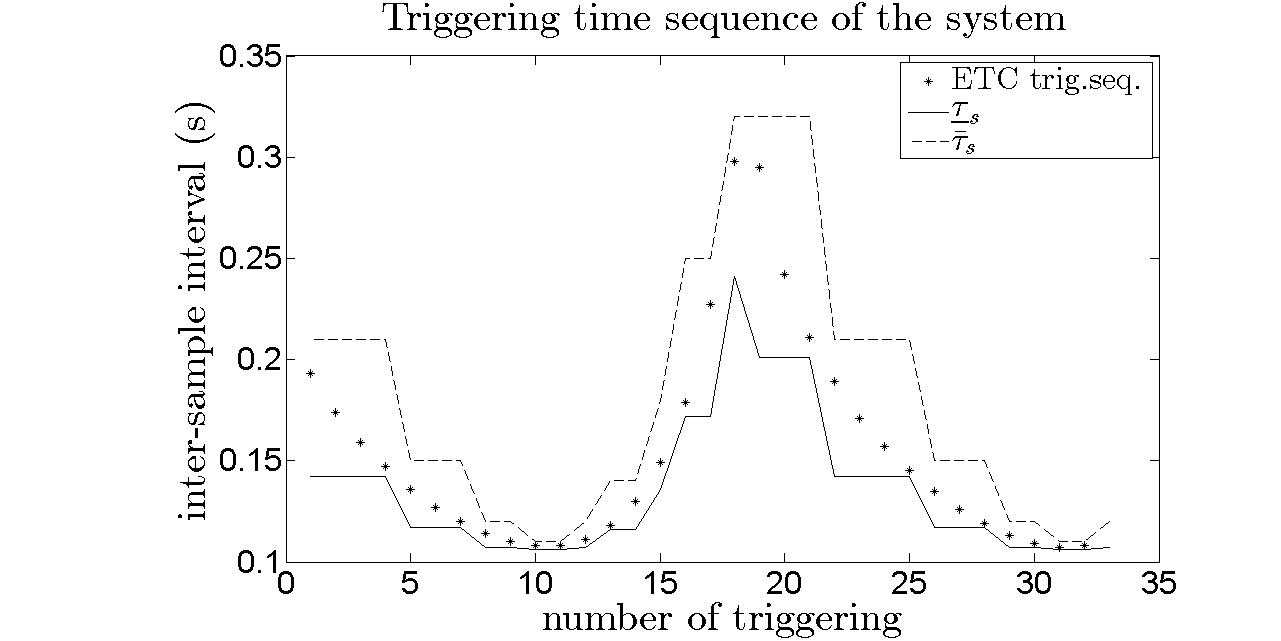}
 \caption{Schematic representation of the fact that the triggering sequence induced by simulation of the the event-triggered control system (asterisks) is contained in the theoretical lower and upper bounds (solid and dashed lines, respectively).} 
 \label{etc_ex1_3}
\end{figure}

In Figure \ref{etc_ex1_6}, a schematic representation of the discrete transition of the resulting TSA is provided. Each asterisk represents a transition from mode $i$ (or $\mathcal{R}_i$) to mode $j$ (or $\mathcal{R}_j$) where $i$ and $j$ are the horizontal and vertical coordinates of the asterisk, respectively.

\begin{figure}[h]
 \centering
  \includegraphics[width=1.1
  \linewidth]{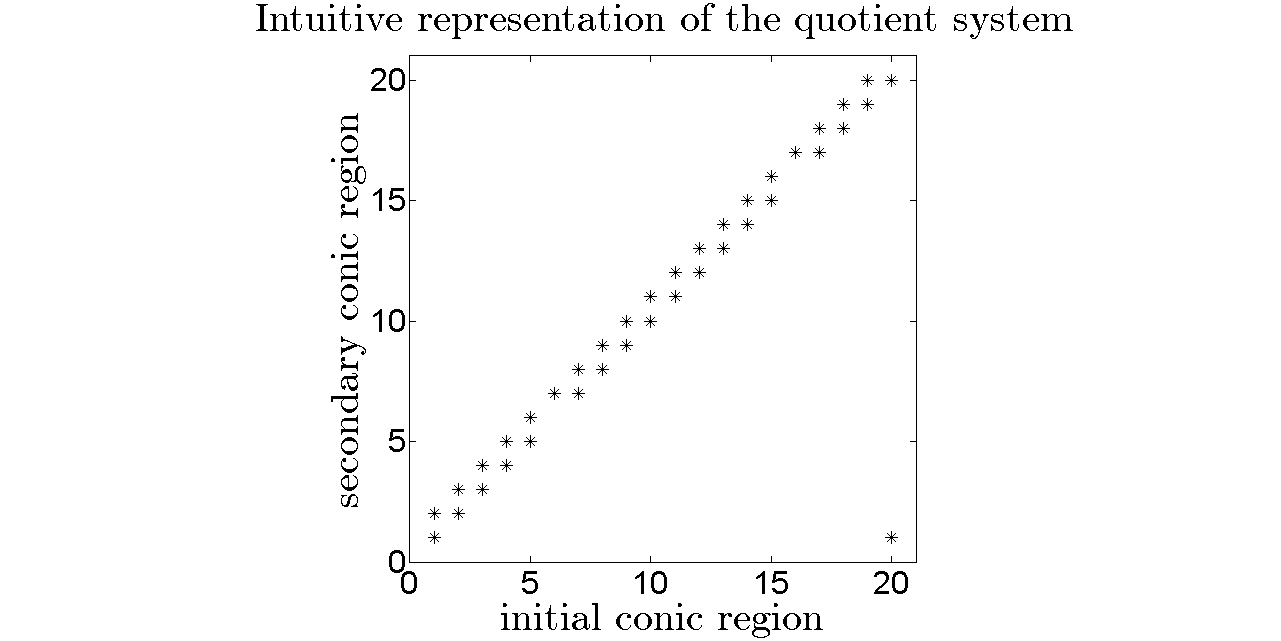}
 \caption{Schematic representation of the transitions in the hybrid automaton (or the timed automaton).} 
 \label{etc_ex1_6}
\end{figure}

\section{Conclusion}
\label{sec_7} 
We have proposed an LMI-based approach which combined with a reachability analysis allows to derive a timed safety automaton that captures the sampling behavior of LTI systems controlled with an event-triggered control strategy. It has been shown that there exists an $\epsilon$-approximate simulation relation from the control system to the derived timed-automaton. These results provide a theoretical basis for automatic schedulability analysis and scheduler synthesis. Thus, an interesting follow-up research direction is to exploit the already existing theory and synthesis tools for timed-automata to synthesize schedulers for event-triggered control loops. This is a significant contribution to the field of event-triggered control in terms of applicability. Until now the only possible schedulability check for a set of event-triggered control loops was limited to a conservative one employing the minimum inter-event interval as descriptor of each of the loops. In turn, such conservative analysis challenged the reasons to implement control loops in the event-triggered fashion versus the traditional periodic approach, as the advantages on bandwidth usage could not be catered to the network infrastructure. With our present results a less conservative check of schedulability is enabled allowing to effectively take advantage of the bandwidth savings on event-triggered implementations. Furthermore, given the direct trade-off between performance and bandwidth usage that the triggering mechanism provides, one can envision more flexible implementations of event-triggered control loops that gently adjust their performance to accommodate to network congestion. 
Additional follow-up lines of research are the extension of this abstraction results to systems affected by disturbances and to classes of non-linear systems. Finally, it seems also natural to explore the connection between the constructed abstractions in the form of TSAs and alternative abstractions employed to model aperiodic traffic in the literature such as the arrival curves employed on network and real-time calculus~\cite{thiele00}. Establishing the link between those two will enable the use of tools from real-time calculus to appropriately dimension and design networks for systems involving event-based control loops.

\appendix  
%


\textbf{Proof of Lemma \ref{lem3}:} We provide a constructive proof in four steps.

\textit{First:} divide the time interval $[0,\bar{\sigma}]$ into $l$ subintervals. The aim of this step is to make the preparations to compute a precise estimation of $\Phi(\cdot)$ by building $l$ small convex embeddings around it instead of building an overly conservative single embedding. 
{For every time instant $\sigma \in [0,\munderbar{\tau}_s]$, there exists $j \in \{0,\dots,\lfloor\frac{\munderbar{\tau}_s l}{\bar{\sigma}} \rfloor \}$ such that $j\frac{\bar{\sigma}}{l} \leq \sigma \leq (j+1) \frac{\bar{\sigma}}{l}$. Define $\sigma '=\sigma-j \frac{\bar{\sigma}}{l}$ ($\sigma' \in [0,\chi]$, with $\chi =\frac{\bar{\sigma}}{l}$ for $j<\lfloor \frac{\munderbar{\tau}_s l}{\bar{\sigma}} \rfloor$ and $\chi=\munderbar{\tau}_s-j\frac{\bar{\sigma}}{l}$ otherwise).}

\textit{Second:} compute a polynomial approximation of $\Phi(\cdot)$ over each subinterval. Employ the relation $\int _0^{a+b} e^{Ar}dr=\int_0^{a}e^{Ar}dr+\int_0^b e^{Ar}dr(A\int_0^{a} e^{Ar}dr+I)$ to simplify $\Lambda(\sigma)$ given in (\ref{my_etc16}) as:

\begin{equation}
\begin{array}{l}
\Lambda(\sigma) =I+M_j(A+BK)+ \int_0^{\sigma '} e^{Ar}dr N_j (A+BK)
\end{array}
\label{my_etc26}
\end{equation}
with $M_j$ and $N_j$ as in (\ref{my_etc23}). Then, by defining two new matrices $\Pi_{1,j}$  and $\Pi_{2,j}$ as in (\ref{my_etc22}), equation (\ref{my_etc26}) can be rewritten as:

\begin{equation}
\begin{array}{l}
\Lambda(\sigma)=\Pi_{1,j}+\int_0^{\sigma '} e^{Ar}dr \Pi_{2,j}.
\end{array}
\label{my_etc29}
\end{equation}

Replace $\int_0^{\sigma '}e^{Ar}dr$ by its $N_{conv}-th$ order Taylor series expansion to approximate $\Phi$ given by (\ref{my_etc18}), i.e.:  

\begin{equation}
\begin{array}{l}
\int_0^{\sigma '}e^{Ar}dr \simeq \sum_{i=1}^{N_{conv}} \frac{A^{i-1}}{i!} \sigma^{'i}.
\end{array}
\label{my_etc30}
\end{equation}

Remember that $N_{conv}+1$ is the number of vertices we consider for polytopic embedding according to time. The Taylor series expansion of $\Phi$ is given by $\sum_{k=0}^{\infty} L_{k,j} \sigma'^k$ with $L_{k,j}$ defined in (\ref{my_etc21}). 

\textit{Third:} bound the error introduced by the $N_{conv}$-th order approximation of $\Phi$ with a constant term. One can derive the $N_{conv}$-th order approximation of $\Phi$ on the time interval $[j\frac{\bar{\sigma}}{l},(j+1)\frac{\bar{\sigma}}{l}]$ using $\tilde{\Phi}_{N_{conv},j}(\sigma')$ given in (\ref{my_etc25}). Denote by $R_{N_{conv},j}(\sigma ')=\Phi(\sigma)-\tilde{\Phi}_{N_{conv},j}(\sigma ')$. We seek to compute a constant scalar $\munderbar{\nu}$ independent of $\sigma '$ such that $R_{N_{conv},j} (\sigma ') \preceq \munderbar{\nu} I$, to establish $x^T \Phi(\sigma) x \leq 0$ from $x^T (\tilde{\Phi}_{N_{conv},j}(\sigma')+ \munderbar{\nu} I) x \leq 0$. Since $R_{N_{conv},j}$ is symmetric $R_{N_{conv},j} (\sigma ') \preceq \lambda_{\max}(\sigma')I$, where $\lambda_{\max}(\sigma')$ is the maximal eigenvalue of $R_{N_{conv},j} (\sigma ')$, and thus $\munderbar{\nu}$ is provided by (\ref{my_etc24}).

\textit{Fourth:} build a convex polytope that contains the matrix exponential function $\tilde{\Phi}_{N_{conv},j}+\munderbar{\nu} I: [0,\chi] \rightarrow \mathcal{M}_n(\mathbb{R})$, using the method proposed in \cite{Het}. From~\cite{Het} we know that if $x^T \munderbar{\Phi}_{(i,j),s} x \leq 0,\; \forall (i,j) \in \mathcal{K}_s=(\{0,\dots,N_{conv} \} \times \{0,\dots,\lfloor \frac{\munderbar{\tau}_s l}{\bar{\sigma}} \rfloor \})$, with $\munderbar{\Phi}_{(i,j),s} = \sum_{k=0}^i L_{k,j} \chi^{k}+\munderbar{\nu} I$, the following holds $x^{T} (\tilde{\Phi}_{N_{conv},j}(\sigma ')+\munderbar{\nu} I) x \leq 0$ and as a result $x^{T} \Phi(\sigma) x \leq 0,\; \forall \sigma \in [0,\munderbar{\tau}_s]$. $\Box$


\textbf{Proof of Lemma \ref{lemm_u1}:} The proof of Lemma \ref{lemm_u1} is analogous to the proof of Lemma \ref{lem3} with the following changes. 
In the first step now: for every $\sigma \in [\bar{\tau}_s, \bar{\sigma}]$, $\exists j \in \{\lfloor\frac{\bar{\tau}_s l}{\bar{\sigma}} \rfloor ,\dots,l-1\}$ such that $j\frac{\bar{\sigma}}{l} \leq \sigma \leq (j+1) \frac{\bar{\sigma}}{l}$, we define $\sigma '=\sigma-j \frac{\bar{\sigma}}{l}$ ($\sigma' \in [0,\chi]$, with  $\chi=(j+1)\frac{\bar{\sigma}}{l}-\bar{\tau}_s$ for $j=\lfloor \frac{\bar{\tau}_s l}{\bar{\sigma}} \rfloor$ and $\chi =\frac{\bar{\sigma}}{l}$ for $j>\lfloor \frac{\bar{\tau}_s l}{\bar{\sigma}} \rfloor$).
In the third step, now we seek to compute instead a constant scalar $\bar{\nu}$ independent of $\sigma'$ such that $R_{N_{conv},j} \succeq \bar{\nu}I$, to establish $x^T\Phi(\sigma)x\geq0$ from $x^T (\tilde{\Phi}_{N_{conv},j}(\sigma')+ \bar{\nu} I) x \geq 0$. Since $R_{N_{conv},j}$ is symmetric, it follows $R_{N_{conv},j} (\sigma ') \succeq \lambda_{\min}(\sigma')I$, where $\lambda_{\min}(\sigma')$ is the minimal eigenvalue of $R_{N_{conv},j} (\sigma')$. This constant can be computed now from~(\ref{my_etc24_1}).
In the fourth step, applying again the results from~\cite{Het}, 
one has that $x^T \bar{\Phi}_{(i,j),s} x \geq 0,\; \forall (i,j) \in \mathcal{K}_s=(\{0,\dots,N_{conv} \} \times \{\lfloor \frac{\bar{\tau}_s l}{\bar{\sigma}} \rfloor,\dots,l-1 \})$, with $\bar{\Phi}_{(i,j),s} = \sum_{k=0}^i L_{k,j} \chi^{k}+\bar{\nu} I$, implies $x^{T} (\tilde{\Phi}_{N_{conv},j}(\sigma ')+\bar{\nu} I) x \geq 0$ and consequently $x^{T} \Phi(\sigma) x \geq 0,\; \forall \sigma \in [\bar{\tau}_s,\bar{\sigma}]$. $\Box$




\ifCLASSOPTIONcaptionsoff
  \newpage
\fi



\bibliographystyle{IEEEtran}
\bibliography{bare_jrnl}

\begin{thebibliography}{10}
\providecommand{\url}[1]{#1}
\csname url@samestyle\endcsname
\providecommand{\newblock}{\relax}
\providecommand{\bibinfo}[2]{#2}
\providecommand{\BIBentrySTDinterwordspacing}{\spaceskip=0pt\relax}
\providecommand{\BIBentryALTinterwordstretchfactor}{4}
\providecommand{\BIBentryALTinterwordspacing}{\spaceskip=\fontdimen2\font plus
\BIBentryALTinterwordstretchfactor\fontdimen3\font minus
  \fontdimen4\font\relax}
\providecommand{\BIBforeignlanguage}[2]{{%
\expandafter\ifx\csname l@#1\endcsname\relax
\typeout{** WARNING: IEEEtran.bst: No hyphenation pattern has been}%
\typeout{** loaded for the language `#1'. Using the pattern for}%
\typeout{** the default language instead.}%
\else
\language=\csname l@#1\endcsname
\fi
#2}}
\providecommand{\BIBdecl}{\relax}
\BIBdecl

\bibitem{EDC}
W.~Heemels, K.~Johansson, and P.~Tabuada, ``An introduction to event-triggered
  and self-triggered control,'' in \emph{Proceedings of the 51st Annual
  Conference on Decision and Control}, Dec 2012, pp. 3270--3285.

\bibitem{Hri}
D.~Hristu-Varsakelis and P.~Kumar, ``Interrupt-based feedback control over a
  shared communication medium,'' in \emph{Proceedings of the 41st IEEE
  Conference on Decision and Control}, vol.~3, Dec 2002, pp. 3223--3228.

\bibitem{Ast}
K.~Astrom and B.~Bernhardsson, ``Comparison of riemann and lebesgue sampling
  for first order stochastic systems,'' in \emph{Proceedings of the 41st IEEE
  Conference on Decision and Control}, vol.~2, Dec 2002, pp. 2011--2016.

\bibitem{Vou}
P.~Voulgaris, ``Control of asynchronous sampled data systems,'' \emph{IEEE
  Transactions on Automatic Control}, vol.~39, no.~7, pp. 1451--1455, 1994.

\bibitem{Tab}
P.~Tabuada and X.~Wang, ``Preliminary results on state-trigered scheduling of
  stabilizing control tasks,'' in \emph{Proceedings of the 45th IEEE Conference
  on Decision and Control}, Dec 2006, pp. 282--287.

\bibitem{Kof}
E.~Kofman and J.~Braslavsky, ``Level crossing sampling in feedback
  stabilization under data-rate constraints,'' in \emph{Proceedings of the 45th
  IEEE Conference on Decision and Control}, Dec 2006, pp. 4423--4428.

\bibitem{Vel}
M.~Velasco, J.~M. Fuertes, and P.~Marti, ``The self triggered task model for
  real-time control systems,'' in \emph{Proceedings of the 24th IEEE Symposium
  on Real-Time Systems (work in progress)}, 2003, pp. 67--70.

\bibitem{Anta}
A.~Anta and P.~Tabuada, ``Self-triggered stabilization of homogeneous control
  systems,'' in \emph{Proceedings of American Control Conference}, June 2008,
  pp. 4129--4134.

\bibitem{Man}
M.~Mazo~Jr., A.~Anta, and P.~Tabuada, ``An {ISS} self-triggered implementation
  of linear controllers,'' \emph{Automatica}, vol.~46, no.~8, pp. 1310 -- 1314,
  2010.

\bibitem{Wang}
X.~Wang and M.~Lemmon, ``Self-triggered feedback systems with state-independent
  disturbances,'' in \emph{Proceedings of American Control Conference}, June
  2009, pp. 3842--3847.

\bibitem{Don}
M.~Donkers and W.~Heemels, ``Output-based event-triggered control with
  guaranteed $\mathcal{L_{\infty}}$-gain and improved and decentralized
  event-triggering,'' \emph{IEEE Transactions on Automatic Control}, vol.~57,
  no.~6, pp. 1362--1376, 2012.

\bibitem{But}
G.~Buttazzo, G.~Lipari, and L.~Abeni, ``Elastic task model for adaptive rate
  control,'' in \emph{Proceedings of The 19th IEEE Symposium on Real-Time
  Systems}, Dec 1998, pp. 286--295.

\bibitem{Cac}
M.~Caccamo, G.~Buttazzo, and L.~Sha, ``Elastic feedback control,'' in
  \emph{Proceedings of the 12th Euromicro Conference on Real-Time Systems},
  2000, pp. 121--128.

\bibitem{Lu}
C.~Lu, J.~A. Stankovic, S.~H. Son, and G.~Tao, ``Feedback control real-time
  scheduling: Framework, modeling, and algorithms,'' \emph{Real-Time Syst.},
  vol.~23, no. 1/2, pp. 85--126, 2002.

\bibitem{Cer}
A.~Cervin and J.~Eker, ``Control-scheduling codesign of real-time systems: The
  control server approach,'' \emph{Embedded Comput.}, vol.~1, no.~2, pp.
  209--224, 2005.

\bibitem{Bhat}
R.~Bhattacharya and G.~Balas, ``Anytime control algorithm: Model reduction
  approach,'' \emph{Journal of Guidance, Control, and Dynamics}, vol.~27,
  no.~5, pp. 767--776, 2004.

\bibitem{Font}
D.~Fontanelli, L.~Greco, and A.~Bicchi, ``Anytime control algorithms for
  embedded real-time systems,'' in \emph{Hybrid Systems: Computation and
  Control}, ser. Lecture Notes in Computer Science, M.~Egerstedt and B.~Mishra,
  Eds.\hskip 1em plus 0.5em minus 0.4em\relax Springer Berlin Heidelberg, 2008,
  vol. 4981, pp. 158--171.

\bibitem{Areq1}
S.~Al-Areqi, D.~Gorges, S.~Reimann, and S.~Liu, ``Event-based control and
  scheduling codesign of networked embedded control systems,'' in
  \emph{Proceedings of American Control Conference}, June 2013, pp. 5299--5304.

\bibitem{Areq2}
S.~Al-Areqi, D.~Gorges, and S.~Liu, ``Stochastic event-based control and
  scheduling of large-scale networked control systems,'' in \emph{Proceedings
  of European Control Conference}, June 2014, pp. 2316--2321.

\bibitem{Alu}
R.~Alur and D.~L. Dill, ``A theory of timed automata,'' \emph{Theoretical
  Computer Science}, vol. 126, pp. 183--235, 1994.

\bibitem{Tab_etc}
P.~Tabuada, ``Event-triggered real-time scheduling of stabilizing control
  tasks,'' \emph{IEEE Transactions on Automatic Control}, vol.~52, no.~9, pp.
  1680--1685, 2007.

\bibitem{Fit}
C.~Fiter, L.~Hetel, W.~Perruquetti, and J.-P. Richard, ``A state dependent
  sampling for linear state feedback,'' \emph{Automatica}, vol.~48, no.~8, pp.
  1860--1867, 2012.

\bibitem{Khal}
H.~K. Khalil, \emph{Nonlinear systems}.\hskip 1em plus 0.5em minus 0.4em\relax
  Upper Saddle River, (N.J.): Prentice Hall, 1996.

\bibitem{Kolm}
V.~B. Kolmanovskii and A.~D. Myshkis, \emph{Applied theory of functional
  differential equations}, ser. Mathematics and its applications.\hskip 1em
  plus 0.5em minus 0.4em\relax Dordrecht, Boston: Kluwer Academic Publishers,
  1992.

\bibitem{Frid}
E.~Fridman, A.~Seuret, and J.-P. Richard, ``Robust sampled-data stabilization
  of linear systems: an input delay approach,'' \emph{Automatica}, vol.~40,
  no.~8, pp. 1441--1446, 2004.

\bibitem{Het}
L.~Hetel, J.~Daafouz, and C.~Lung, ``Lmi control design for a class of
  exponential uncertain systems with application to network controlled switched
  systems,'' in \emph{Proceedings of American Control Conference}, July 2007,
  pp. 1401--1406.

\bibitem{Chut}
A.~Chutinan and B.~Krogh, ``Computing polyhedral approximations to flow pipes
  for dynamic systems,'' in \emph{Proceedings of the 37th IEEE Conference on
  Decision and Control}, vol.~2, Dec 1998, pp. 2089--2094.

\bibitem{Henz94}
T.~A. Henzinger, X.~Nicollin, J.~Sifakis, and S.~Yovine, ``Symbolic model
  checking for real-time systems,'' \emph{Information and computation}, vol.
  111, no.~2, pp. 193--244, 1994.

\bibitem{Ewald}
G.~Ewald, \emph{Combinatorial convexity and algebraic geometry}, ser. Graduate
  texts in mathematics.\hskip 1em plus 0.5em minus 0.4em\relax New York:
  Springer, 1996.

\bibitem{Abs}
P.~Tabuada, \emph{Verification and Control of Hybrid Systems: A Symbolic
  Approach}, 1st~ed.\hskip 1em plus 0.5em minus 0.4em\relax Springer, 2009.

\bibitem{Bengtsson04}
J.~Bengtsson and W.~Yi, ``Timed automata: Semantics, algorithms and tools,'' in
  \emph{Lectures on Concurrency and Petri Nets}.\hskip 1em plus 0.5em minus
  0.4em\relax Springer, 2004, pp. 87--124.

\bibitem{Deca}
R.~DeCarlo, M.~Branicky, S.~Pettersson, and B.~Lennartson, ``Perspectives and
  results on the stability and stabilizability of hybrid systems,''
  \emph{Proceedings of the IEEE}, vol.~88, no.~7, pp. 1069--1082, 2000.

\bibitem{Alur}
R.~Alur and R.~P. Kurshan, ``Timing analysis in cospan,'' in \emph{In Hybrid
  Systems III}.\hskip 1em plus 0.5em minus 0.4em\relax Springer-Verlag, 1996,
  pp. 220--231.

\bibitem{Upp}
K.~G. Larsen, P.~Pettersson, and W.~Yi, ``Uppaal in a nutshell,'' \emph{Int.
  Journal on Software Tools for Technology Transfer}, vol.~1, pp. 134--152,
  1997.

\bibitem{thiele00}
L.~Thiele, S.~Chakraborty, and M.~Naedele, ``Real-time calculus for scheduling
  hard real-time systems,'' in \emph{Proceedings of the IEEE International
  Symposium on Circuits and Systems}, vol.~4, 2000, pp. 101--104.

\end{thebibliography}
\end{document}